\documentclass [orivec] {llncs}

\usepackage{color}
\usepackage{url}
\usepackage{graphicx}
\usepackage{url}
\usepackage{placeins}
\usepackage{mathptmx}
\usepackage{euscript}
\usepackage{pgf}
\usepackage{tikz}
\usetikzlibrary{automata}
\usetikzlibrary{arrows,petri}
\usepackage[applemac]{inputenc}
\usepackage{algorithm}
\usepackage{algpseudocode}
\usepackage{epsfig}
\usepackage{subfigure}
\usepackage{multirow}
\usepackage{graphicx}
\usepackage{color}
\usepackage{amsmath,amssymb,accents}

\usepackage{multirow}

\newcommand{\Gstar}{\accentset{\star}{G}}

\newcommand{\Pstar}{\accentset{\star}{P}}

\newcommand{\Tstar}{\accentset{\star}{T}}
\newcommand{\Tring}{\accentset{\circ}{T}}

\newcommand{\BN}{[N]}

\algtext*{EndWhile}
\algtext*{EndIf}
\algtext*{EndFor}

\newcommand{\fireseq}[1]{{\ensuremath{[{#1}\rangle}}}

\newcommand{\si}{{\ensuremath{\sigma}}}        

\newcommand{\mk}{\mathbf{m}}

\newcommand{\marc}{\textup{\textbf{\textit m}}}

\begin{document}

\title{Analysis of Petri Net Models through Stochastic Differential Equations}

\author{Marco Beccuti\inst{1}, Enrico Bibbona\inst{2},  Andras Horvath\inst{1}, Roberta Sirovich\inst{2},  Alessio Angius\inst{1} and Gianfranco Balbo\inst{1}}

\institute{%
Universit\`a di Torino,
Dipartimento di Informatica\\
\email{\{beccuti,angius,horvath,balbo\}@di.unito.it}
\and
Universit\`a di Torino,
Dipartimento di Matematica\\
\email{\{roberta.sirovich,enrico.bibbona\}@unito.it}
}

\maketitle


\begin{abstract}


It is well known, mainly because of the work of Kurtz, that density
dependent Markov chains can be approximated by  sets of ordinary
differential equations (ODEs) when their indexing parameter grows very large.
This approximation cannot capture the stochastic nature of the process and,
consequently, it can provide an erroneous view of the behavior of the
Markov chain if the indexing parameter  is not sufficiently high.  Important
phenomena that cannot be revealed include non-negligible variance and
bi-modal population distributions.  A less-known approximation proposed by
Kurtz applies stochastic differential equations (SDEs) and provides
information about the stochastic nature of the process.

In this paper we apply and extend this diffusion approximation to study stochastic
Petri nets.  We identify a class of nets whose underlying stochastic
process is a density dependent Markov chain whose indexing parameter is a
multiplicative constant which identifies the population level expressed by the initial marking
and we provide means to automatically construct
the associated set of SDEs.
Since the diffusion
approximation of Kurtz considers the process only up to the time when it
first exits an open interval, we extend the approximation by a machinery
that
mimics the behavior of the Markov chain at the boundary and allows thus
to apply the approach to
a wider set of problems. The resulting process is of the jump-diffusion
type. We illustrate by examples that the jump-diffusion
approximation which extends to bounded domains can be much more
informative than that based on ODEs as
it can provide accurate quantity distributions even when they are multi-modal
and even for relatively small population levels. Moreover, we show that the method is
faster than
simulating the original Markov chain.

\end{abstract}


\section{Introduction}
Stochastic Petri Nets (SPNs) are a well-known formalism widely used for the
performance analysis of complex Discrete Event Dynamic Systems
\cite{BOABCDF95,Bal-01}.  The advantages of modeling with SPNs include their
well defined time semantics which often allows the direct definition of the
Continuous Time Markov Chain (CTMC) that represents the SPN's underlying stochastic
process whose state space is isomorphic to the reachability set of the net.
The analysis of real  systems often requires the construction of SPN
models with huge state spaces that may hamper the practical relevance of the
formalism and have motivated the development of many techniques capable
of reducing the impact of this problem.  However, when the model includes
large groups of elements (e.g., Internet users,  human populations,
molecule quantities) most of these techniques may turn out to be
insufficient, so that expected values and probability distributions must be estimated with
Discrete Event Simulation \cite{Fish1978,GaetaTSE96}.
An alternative to simulation is the approximation of the stochastic model with a deterministic one in which its
time evolution is represented with a set of Ordinary
Differential Equations (ODEs) whose solution is interpreted as the
approximate expected value of the quantities of interest.

The convergence of the solution of the system of ODEs to the expected
values of their corresponding quantities, when the sizes of the involved
populations grow very large has been the subject of many papers.  Most of
these are based on the work of Kurtz~\cite{Ku70} and have shown that the
accuracy of its approximation is acceptable when the model represents a
system of large interacting population quantities \cite{Tr10,TrGiHi12}.
Unfortunately, there are many cases in which the deterministic
approximation is not satisfactory because the obtained approximate expected
values give little or even erroneous information about the actual
population levels.  This happens either when  the population sizes are not large
enough to rule out the variability of the process and the obtained expected
values are not reliable, or when the population distributions are multi-modal, a case in which the mean does not provides much
information~\cite{JANE}.

In this paper we extend a stochastic approximation of the CTMC that has been introduced in \cite{kurtz1976limit} by Kurtz.
There, the fluidisation is augmented by a suitable noise term
which accounts for the stochasticity of the original system yielding a system of Stochastic Differential Equations
(SDEs). Unfortunately the proposed diffusion approximation is only valid up to the first time the system exits a suitable open domain. In real systems however
the boundaries of the state space are repeatedly visited (the buffer of a queue may get full,
all resources of a system may be in use) and the system can stay there for a finite time and come back to the interior again and again. In these cases, the diffusion approximation proposed in \cite{kurtz1976limit} gives a totally incomplete description. We propose an improved approximation which results in a jump diffusion process that visits the boundaries and with jumps that push back the process to
the interior of its state space.

The paper is organized as follows.  In Section 2, in order to provide the
theoretical background of our exposition, we review some of Kurtz's results
regarding deterministic and diffusion approximations.  In Section 3 we
apply these results to SPNs and extend Kurtz's results to treat the
barriers of the state space.  Numerical experiments are provided in Section
4 and conclusions are drawn in Section 5.

\newcommand{\comAH}[1]{{\color{blue}{AH: #1}}}

\section{Two fluid approximations for density dependent CTMCs}

In this section we give a brief overview of two possible approximations of
density dependent CTMCs.  Both of these are ``fluid'' in the sense that the
state space of the approximating process is continuous.  The first one is
the well-known deterministic fluid limit introduced in \cite{Ku70} which
employs a set of ordinary differential equations.  The second one is the
so-called diffusion approximation introduced in \cite{kurtz1976limit} which
uses SDEs.  The decisive difference between
the two approximations is that, at any time, the first
one provides a single number per random variable of interest (e.g., number of customers in service, number of molecules in a cell), which is usually interpreted
as the approximate expected value, while the second one leads
to an approximate joint distribution of all the variables of interest.

\subsection{Density dependent CTMCs\label{sec:ddCTMC}}

In the following we will denote $\mathbb{R}$, $\mathbb{Z}$ and $\mathbb{N}$
the set of real, integer and natural numbers, respectively. Given a 
positive constant, $r$, we will denote by $\mathbb{R}^{r}$ the $r$--dimensional cartesian
product of the space $\mathbb{R}$. The letter $u$ will be dedicated to the
time index ranging continuously between $[0, +\infty)$ or $[0,T]$ when
  specified. The discrete states of a continuous time Markov chain will be
  denoted as $k$ or $h$ and range in the state space that is included in
  $\mathbb{Z}^{r}$. We will always consider the abstract probability space
  to be given as $(\Omega,\mathcal{F},\mathbb{P})$, where $\mathbb{P}$ is
  the probability measure. Furthermore, $\mathbb{E}$ will denote the
  expectation with respect to $\mathbb{P}$.

\begin{definition}\label{def:density_dependent}%
A family of Markov chains $X^{\BN}(u)$ with indexing parameter $N$ and with state
space $S^{\BN} \subseteq \mathbb{Z}^{r}$, is called \emph{density dependent} iff there
exists a continuous non-zero function $f: \mathbb{R}^r \times \mathbb{Z}^{r} \rightarrow \mathbb{R}$ such that the instantaneous transition rate (intensity) from state $k$
to state $k+l$ can be written as
\begin{equation}\label{eq:form}
q^{\BN}_{k,k+l}=N f\left(\frac{k}{N},l\right), \quad l \neq 0.
\end{equation}
\end{definition}
In the previous definition, the first argument of the function $f$ can be seen as a
normalized state (with respect to the indexing parameter $N$) of the CTMC and the second argument as a vector that
describes the effect of a transition (change of state).  Consequently, eq. (\ref{eq:form})
states that, given a vector $l$, the intensities depend on the
normalized state, $k/N$, and are proportional to the indexing parameter
$N$.  In the following we denote the set of possible state changes by $C$,
i.e., $C=\{l:l \in \mathbb{Z}^{r}, l \neq 0, q^{\BN}_{k,k+l}\neq 0\}$.

Some important models do not satisfy Definition \ref{def:density_dependent} exactly but are still treatable in the same framework. For this reason we introduce the following more general definition.
\begin{definition}\label{def:nearly-density_dependent}%
A family of Markov chains $X^{\BN}(u)$ with parameter $N$ and with state
space $S^{\BN} \subseteq \mathbb{Z}^{r}$, is called \emph{nearly density dependent} iff there
exists a continuous non-zero function $f: \mathbb{R}^r \times \mathbb{Z}^{r} \rightarrow \mathbb{R}$ such that the instantaneous transition rate (intensity) from state $k$
to state $k+l$ can be written as:
\begin{equation}\label{eq:qnearly}
q^{\BN}_{k,k+l}=N \left[ f\left(\frac{k}{N},l\right) + O\left( \frac{1}{N} \right) \right], \quad l \neq 0.
\end{equation}
\end{definition}
\begin{example}\label{ex:ddctmc}
As an example consider a closed network of $r$ infinite server queues with
exponential service time distributions and with $N$ jobs circulating in it.
The state space is thus $S^{\BN}=\{k: k \in \mathbb{Z}^{r}, 0\leq k_i \leq N,
\sum_i k_i=N\}$. Let $\mu_i$ denote the service intensity of the $i$th
queue, $r_{i,j}$ the routing probabilities with $r_{i,i}=0$, and $l_{i,j},
i\neq j,$ a vector with $-1$ in position $i$, $+1$ in position $j$ and 0s
elsewhere.  Then the transition rates and the functions $f$ are

\begin{align}
q^{\BN}_{k,k+l}=\left\{
\begin{array}{ll}
k_i \mu_i r_{i,j} & \mbox{if}~l=l_{i,j} \\
0 & \mbox{otherwise}
\end{array}
\right.
& \quad \quad
f(y,l)=\left\{
\begin{array}{ll}
y_i \mu_i r_{i,j} & \mbox{if}~l=l_{i,j} \\
0 & \mbox{otherwise}
\end{array}
\right.
\end{align}
and thus that the transition rates are in the form given in eq. \eqref{eq:form}, with $y=k/N$.  For
this simple example the transition rates depend only on one component of
the state.  This is not a necessary condition for a CTMC to be
density dependent and in Section~\ref{sec:exp_results} we show examples for which this
condition does not hold.
\end{example}

In order to gain a better understanding of the property of density
dependence, let us introduce some general concepts from the theory of
Markov chains. Among the many books devoted to this topic, we refer the
reader to \cite{bremaud1999}.

For a general Markov chain $M(u)$ with state space $S \subseteq  \mathbb{Z}^{r}$ and instantaneous transition rates $q_{k,h}$, let us introduce the following key object
\begin{align}\label{eq:F}
F_{M}(k) = \sum_{h\in S} (h-k) q_{k,h} \quad k \in S.
\end{align}
The function $F_{M}$ will be referred to as the \emph{generator} of the chain.
Under suitable hypothesis the expectation of $M(u)$ solves the following \emph{Dynkin equation} cf. \cite[Chapter 9, Theorem 2.2]{bremaud1999}
\begin{align}\label{eq:Dynkin}
\frac{d \mathbb{E}[M(u)]}{du} = \mathbb{E}[F_{M}(M(u))].
\end{align}
\begin{example}
For the density dependent CTMC introduced in Example~\eqref{ex:ddctmc}, the
$m$th entry of the generator is
\[
(F_{X^{\BN}}(k))_m=\sum_{i=1,i\neq m}^r k_i \mu_i r_{i,m} - k_m \mu_m
\]
Define by $\pi_k(u)$ the transient probabilities and apply the
Dynkin equation.  We obtain
\begin{align}
\notag
\frac{d ( \mathbb{E}[M(u)])_m}{du}=&
\sum_{k \in E_N} \pi_k(u) \left(\sum_{i=1,i\neq m}^r k_i \mu_i r_{i,m} - k_m \mu_m \right)
\\ \notag
=&\sum_{i=1,i\neq m}^r ( \mathbb{E}[M(u)])_i \mu_i r_{i,m} - ( \mathbb{E}[M(u)])_m \mu_m
\end{align}
which provides a set of ODEs that can be used to calculate the mean queue
length of each queue.  Note that in general it is not the case that
applying the Dynkin equation leads to such ODEs from which the mean
quantities can be directly obtained.
\end{example}

In order to bring to the same scale the state spaces of all the CTMCs
$X^{\BN}(u)$, it may be convenient to
introduce the family of normalized CTMCs as $Z^{\BN}(u)=\frac{X^{\BN}(u)}{N}$,
which is also referred to as the density process. Notice that the density process will have state space $\left\{ \frac{k}{N}, k \in \mathbb{Z}^{r} \right\}$. Two invariance
properties of density dependent CTMCs can  then be stated using the
normalized chains.
\begin{property}\label{rmk:genC}
If the set of possible state changes $C$ does not depend on $N$, the density dependence property of the family $X^{\BN}(u)$ is equivalent to
require that for the family of the normalized CTMCs, $Z^{\BN}(u)$, the
generator does not depend on $N$. Indeed
\begin{align}
\label{eq:gen_cost}
F_{Z^{\BN}}\left(\frac{k}{N}\right) = \sum_{l \in C} \frac{l}{N} \cdot
p^{\BN}_{\frac{k}{N},\frac{k}{N}+\frac{l}{N}} = \sum_{l \in C} \frac{l}{N} \cdot
q^{\BN}_{k,k+l} = \sum_{l\in C} l f\left(\frac{k}{N},l\right) = F \left(\frac{k}{N}\right)
\end{align}
where $p^{\BN}_{\frac{k}{N},\frac{h}{N}}$ and $q^{\BN}_{k,h}$ are the instantaneous transition
rates of the processes $Z^{\BN}$ and $X^{\BN}$, respectively. Hence, if the increments $l$ are constant with $N$, i.e. the set $C$ does not depend on $N$ and comprises elements which are all independent of $N$, the generator is a function that depends on $N$ only through the state $\frac{k}{N}$.
Let us notice that if the family $X^{\BN}(u)$ is only nearly density dependent then
$F_{Z^{\BN}}\left( \frac{k}{N} \right) = F\left(  \frac{k}{N} \right)+ O\left( \frac{1}{N} \right)$.
\end{property}
\begin{property}
Each element of the family $Z^{\BN}$ solves the same Dynkin equation
\begin{align}\label{eq:dynkinZ}
\frac{d \mathbb{E}[Z^{\BN}(u)]}{du} = \mathbb{E}[F(Z^{\BN}(u))].
\end{align}
\end{property}

The results reported in the following two subsections was demonstrated by
Kurtz exploiting the above properties.

\subsection{From CTMCs to ODEs}\label{sec:ctmc2ode}

In~\cite{Ku70} Kurtz has shown that given a nearly density dependent family of CTMCs $X^{\BN}(u)$, if $\lim_{N \rightarrow \infty} Z^{\BN}(0)=z_{0}$, then, under
relatively mild conditions on the generator $F$ given in eq.\eqref{eq:gen_cost},
the density process $Z^{\BN}$ converges (in a sense to be precised) to a deterministic function $z$
which solves the ODE  \footnote{Eq. \eqref{eq:odeK} is equivalent to the form $\frac{dz(u)}{du}=
F(z(u))$. We have chosen the ``differential'' form written in
\eqref{eq:odeK} to be consistent with the notation that will be introduced
in Section \ref{sec:sde} for the stochastic differential equations. }
\begin{align}\label{eq:odeK}
d z(u) &=F(z(u))du, \quad \quad z(0)=z_{0}.
\end{align}
In~\cite{Ku70} the following convergence in probability is used: for every $\delta >0$
\begin{align}\label{eq:convergenceK}
\lim_{N \rightarrow \infty}
\mathbb{P}\left\{
\sup_{u\leq T}\left|
Z^{\BN}(u)-z(u)
\right|>\delta
\right\}=0.
\end{align}
where $T$ is the upper limit of the finite time horizon.

The function $z(u)$ is usually interpreted as the asymptotic mean of the process as it
solves an equation which is analogous to eq. \eqref{eq:dynkinZ}.  The
difference $Z^{\BN}(u)-z(u)$ can  be interpreted as the ``noisy'' part of
$Z^{\BN}(u)$.  It was shown in \cite{Ku70} that for $N \rightarrow \infty$
the density process $Z^{\BN}(u)$ flattens  at its mean value and that the
magnitude of the noise is
\begin{align}
\label{eq:odediff}
Z^{\BN}(u) - z(u)= O\left( \frac{1}{\sqrt{N}} \right).
\end{align}

The result expressed by eq.\eqref{eq:convergenceK} is often used to approximate the density
dependent process $X^{\BN}(u)=NZ^{\BN}(u)$ with the
deterministic function $x^{\BN}(u)=Nz(u)$, in case of a finite $N$. In doing so, this approximation disregards the
noise term which is now of order $\sqrt{N}$ that is small compared with the
order of the mean (that is $N$), but not in absolute terms. Moreover, it ignores
every details of the probability distribution of $X^{\BN}(u)$ except for the
mean. It is easy to see that there are cases, e.g., multi-modal distributions, where the mean gives too
little information about the  location of the probability mass,
cf \cite{JANE}.

Let us stress that the convergence holds only if $\lim_{N \rightarrow
  \infty} Z^{\BN}(0)=z_{0}$, meaning that the corresponding sequence of
initial conditions $X^{\BN}(0)$ needs to grow linearly
with $N$. In particular if $X^{\BN}(u)$ is multivariate, each component should grow
with the same rate.

\subsection{From CTMCs to SDEs}\label{sec:sde}

An approximation of a density dependent family $X^{\BN}$ which preserves its stochastic nature and 
has a better order of convergence was proposed in
\cite{kurtz1976limit,kurtz1978strong}.
It has been shown in \cite{kurtz1976limit} that, given  an open set $S \subset \mathbb{R}^{r}$,
the density process $Z^{\BN}$ can be approximated by the diffusion process
$Y^{\BN}$ with state space $S$ and solution of the following SDE
\begin{align}\label{eq:sde_ctmc}
dY^{\BN}(u)  =
F(Y^{\BN}(u)) du
 + \sum_{l \in C} \frac{l}{\sqrt{N}}  \sqrt{f(Y^{\BN}(u),l)}\;dW_{l}(u)
\end{align}
where $W_{l}$ are independent standard one-dimensional Brownian
motions and $f$ is given in eq. \eqref{eq:form}. The approximation holds up to the first time $Y^{\BN}$ leaves
$S$. A rigorous mathematical treatment of SDEs can be found in
\cite{klebaner}. In the physical
literature the notation $\frac{dW(u)}{du}=\xi(u)$ is often used even if
Brownian motion is nowhere differentiable and $\xi(u)$ is called a
\emph{gaussian white noise}. This SDE approach that goes back to the
already cited \cite{kurtz1976limit,kurtz1978strong} has been applied in
many contexts, e.g., it is used under the name of \emph{Langevin equations}
to model chemical reactions in \cite{gillespie}.

The structure of eq. \eqref{eq:sde_ctmc} is the following: the first term
is the same that appears in eq.~\eqref{eq:odeK}, while the second term
represents the contribution of the noise and is responsible for the
stochastic nature of the approximating process $Y^{\BN}$.  A further relation
between eq.~\eqref{eq:sde_ctmc} and eq.~\eqref{eq:odeK} can be obtained by
considering that the stochastic part of the equation is proportional to
$1/\sqrt{N}$, meaning that as $N \rightarrow \infty$ this term becomes
negligible and $Y_{\infty}(u)$ solves the same ODE written in
eq. \eqref{eq:odeK}.  Let us remark that the construction of such noise is
not based on an ad hoc assumption, but is  derived from the
structure of the generator of the original Markov chain.

As for the relation between the diffusion approximation and the original density
process, in \cite{kurtz1976limit},  it has been proven that, for any finite
$N$, we have
\begin{align}
Z^{\BN}(u)- Y^{\BN}(u)= O\left( \frac{\log N}{N} \right)
\end{align}
which, compared to eq. \eqref{eq:odediff}, is a better convergence rate.
Thus, the processes $NY^{\BN}(u)$ approximates the density dependent CTMCs
$X^{\BN}(u)$ with an error of order $\log N$, much better than the $\sqrt{N}$
of the deterministic fluid approximation.

Finally, let us stress that the approximation is valid only up to the first
exit time from the open set $S$.  For many applications the natural state
space is bounded and closed and the process may reach the boundary of $S$
in a finite time $\tau$ with non-negligible probability.  In such cases,
since the approximating process $Y^{\BN}(u)$ is no longer defined for any $u
\geq \tau$, this approximation is not applicable.  To overcome
this limitation suitable boundary conditions must be set and this problem,
that was considered neither in \cite{kurtz1976limit} nor in
\cite{kurtz1978strong}, will be tackled in Section~\ref{sec:bd}, representing our main contribution.


\section{From SPNs to fluid approximations}\label{sec:SPNfluid}

In this section we first introduce stochastic Petri nets and give a
condition under which their underlying CTMC is density dependent.  Then, we
reinterpret the results discussed in Section~\ref{sec:ctmc2ode} and
\ref{sec:sde} in terms of SPNs.  Finally, we extend the diffusion
approximation to bounded domains by adding jumps to the diffusion that
mimics the behavior of the original CTMC at the barrier.

\subsection{Density dependent SPNs}

Petri Nets (PNs) are bipartite directed graphs with two types of nodes: places and
transitions.  The places, graphically represented as circles, correspond to
the state variables of the system (e.g.,  number of jobs in a queue), while
the transitions, graphically represented as rectangles, correspond to the
events (e.g., service of a client) that can induce state changes.  The arcs
connecting places to transitions (and vice versa) express the relations
between states and event occurrences.  Places can contain tokens (e.g.,
jobs) drawn as black dots within the places. The state of a PN, called
marking, is defined by the number of tokens in each place.  The evolution
of the system is given by the occurrence of enabled transitions, where a
transition is enabled iff each input place contains a number of tokens
greater or equal than a given threshold defined by the multiplicity of the
corresponding input arc.  A transition occurrence, called firing, removes a
fixed number of tokens from its input places and adds a fixed number of
tokens to its output places (according to the multiplicity of its
input/output arcs).

The set of all the markings that the net can reach, starting from the
initial marking through transition firings, is called the Reachability Set
(RS).  Instead, the dynamic behavior of the net is
described by means of the Reachability Graph (RG), an oriented graph whose
nodes are the markings of the RS and whose arcs represent the transition
firings that produce the corresponding marking changes.

Stochastic Petri Nets (SPNs) are PNs where the firing of each transition is assumed to occur after a delay (firing time) from the time it is enabled. In SPNs these delays are assumed to be random variables with negative exponential distributions  \cite{molloy:spn}. Each transition of an SPN is thus associated with a rate that represents the parameter of its firing delay distribution. Firing rates may be marking dependent. When a marking is entered an exponentially
distributed random delay is sampled for each enabled transition according
to its intensity.  The transition with the lowest delay fires and the
system changes marking accordingly.

Here we recall  the notation and the basic definitions  used in the rest of the paper.
\begin{definition}
A stochastic Petri net (SPN) system is a tuple $\mathcal{N}=( P, T, I, O,
  \mathbf{m}_{0}, \lambda )$:
\begin{itemize}
\item $P=\{p_i\}_{1 \leq i \leq n_{p}}$ is a finite and non empty set of {\em places}.
\item $T=\{t_i\}_{1 \leq i \leq n_{t}}$ is a finite, non empty set of {\em transitions} with $P \cap T=\emptyset$.
\item $I, O: P \times T \rightarrow \mathbb{N}$ are the {\em input},
{\em output} functions that define the arcs of the net and that specify their multiplicities.
\item ${\bf m_0}: P \rightarrow \mathbb{N}$
is a multiset on $P$ representing the {\em initial marking},
\item $\lambda: T  \rightarrow \mathbb{R}$ gives the {\em firing
  intensity} of the transitions.
\end{itemize}
\end{definition}

The overall effect of a transition is described by the function $L=O-I$.
The values assumed by the function $I, O$ and $L$ can be collected in
$n_p \times n_{t}$ matrices (which we still call $I, O$ and $L$) whose
entries are $I(p_{i},t_{j}), O(p_{i},t_{j})$ and $L(p_{i},t_{j})$,
respectively.  By $I(t)$ we denote the column of $I$ corresponding to
transition $t$ (the same holds for $O$ and $L$). The matrix $L=O-I$ is
called the \emph{incidence matrix}.

A {\em marking} (or state) ${\bf m}$ is a function ${\bf m}: P \rightarrow
\mathbb{N}$ identified with a multiset on $P$ which can be seen also as a
vector in $ \mathbb{N}^{n_p}$.  A transition $t$ is \emph{enabled} in
marking $\mk$ iff $\mk(p) \geq I(p, t)$, $\forall p \in P $ where $\mk(p)$
represents the number of tokens in place $p$ in marking $\mk$.  Enabled
transitions may \emph{fire}, so that the firing of transition $t$ in
marking $\mk$ yields a new marking $\mk' = \mk -{I(t)}^T+ {O(t)}^T=\mk + {L(t)}^T$.
Marking $\mk'$ is said to be reachable from $\mk$ because of the firing of
$t$ and is denoted by $\mk\fireseq{t} \mk'$.
The firing  of a \emph{sequence} $\si$ of transitions enabled at $\mk$ and  yielding
$\mk'$ is denoted similarly: $\mk\fireseq{\si}\mk'$.

Let $E(\marc)$ be the set of transitions enabled in marking $\mk$.  The
enabling degree of transition $t$ in marking $\mk$ is defined as
\begin{equation}\label{eq:enabling}
\forall t \in E(\marc) : \; e(t,\marc) = \min_{j:I(p_j,t)\not=0} \left\lfloor \frac{\marc(p_j)}{I(p_j,t)} \right\rfloor.
\end{equation}
which implies that a transition $t \in E(\mk)$ is enabled in marking
$\mk-(e(t,\mk)-1)I(t)^T$, but not in marking $\mk-e(t,\mk)I(t)^T$.  The notion
of enabling degree is particularly useful when a transition intensity is
proportional to the number of tokens in the input places of the
transition, i.e., when the transition models the infinite server mechanism
\cite{BOABCDF95}.

In the CTMC that underlies the behavior of an SPN, states are identified
with markings and a change of marking of the SPN corresponds to a change of
state of the CTMC. If we assume that all the transitions of the SPN use an
infinite server policy, then the transition rate from state $\mk$ to state
$\mk'$ in the CTMC can be written as
\begin{equation}\label{eq:qSPN}
q_{\marc,\marc'} = \sum_{t: L(t)=(\marc'-\marc)}\lambda(t) e(t,\marc)
\end{equation}
where $\lambda(t)$ depends only on the transition and the marking
dependence is determined by the enabling degree $e(t,\marc)$.

A vector $\nu \in \mathbb{N}^{n_p}$ is called a \emph{P--semiflow} of the
SPN if it satisfies $\nu \, L = 0$. All P--semiflows of an SPN can be
obtained as linear combination of the P--semiflows that are elements of a
minimal set. Given a marking $\marc$, the quantities $ \marc \, \nu^{T}$
are invariant and hence equal to $ \marc_{0} \, \nu^{T}$ where $\marc_{0}$
is the initial marking.

\begin{proposition}\label{prop:ddSPN}
Let $\mathcal{N}=( P, T, I, O, \mathbf{m}_{0},\lambda )$ be an SPN model
where all places are covered by  P--semiflows and all transitions use an
infinite server policy. Let us consider the family $\mathcal{N}^{\BN}$ of SPN
models with indexing parameter $N$ obtained from $\mathcal{N}$ by considering an
increasing sequence of initial markings $\marc_{0}^{\BN}= N \alpha$ for a
given vector $\alpha$. The corresponding family of CTMCs is nearly density
dependent and the marking in each place has a bound which grows at most linearly with $N$. If the multiplicity of the input arcs are all unitary, the family is also density dependent.
\end{proposition}
\begin{proof}
Let $\{ \nu^{(\eta)} \}_{\eta=1}^{\kappa}$ denote a minimal set of
P--semiflows for $\mathcal{N}$ and, consequently, for any $\mathcal{N}^{\BN}$
(P--semiflows are independent of the initial marking). For each place the
marking is bounded by
\begin{align}\label{eq:bound}
\marc(p_{j}) \leq \min_{\eta : \nu^{(\eta)}_{j} \neq 0} \left\{ \frac{\marc^{\BN}_{0} \nu^{(\eta) \, T}}{\nu^{(\eta)}_{j}} \right\}= N \min_{\eta : \nu^{(\eta)}_{j} \neq 0} \left\{ \frac{\alpha \nu^{(\eta) \, T}}{\nu^{(\eta)}_{j}} \right\},
\end{align}
where $\nu^{(\eta)}_{j}$ is the component of the P--semiflow vector
$\nu^{(\eta)}$ that corresponds to the place $p_{j}$. By eq. \eqref{eq:qSPN}
and \eqref{eq:enabling}
\begin{align*}
q_{\marc,\marc+l} = \sum_{t: L(t)=l} \lambda(t) \min_{j:I(p_j,t)\not=0} \left\lfloor \frac{\marc(p_j)}{ I(p_j,t)} \right\rfloor.
\end{align*}
For any $j$ we have
\[  \left\lfloor \frac{\marc(p_j)}{ I(p_j,t)} \right\rfloor = \frac{\marc(p_j)}{ I(p_j,t)} - \frac{R_{j}}{I(p_j,t)}, \]
where $R_{j}< I(p_{j},t)$ is the remainder of the division of $\marc(p_{j})$ by $I(p_{j},t)$. So
\begin{align*}
q_{\marc,\marc+l} =&  N \sum_{t: L(t)=l} \lambda(t)  \min_{j:I(p_j,t)\not=0} \left\{  \frac{\marc(p_{j})}{NI(p_{j},t)} - \frac{R_{j}}{N} \right\} \\
= & N \left[ \sum_{t: L(t)=l} \lambda(t)  \min_{j:I(p_j,t)\not=0} \left\{  \frac{\marc(p_{j})}{NI(p_{j},t)} \right\} + O\left( \frac{1}{N} \right) \right].
\end{align*}
By Definition
\ref{def:nearly-density_dependent}, the proposition follows. Notice that when the input arcs are all unitary the remainder is null, $R_{j}=0$ for any $j$, and hence the term disappears and the chain is density dependent.
\end{proof}

Let us notice that the hypothesis $\marc_{0}^{\BN}= N \alpha$ implies that
the initial number of tokens in any place should grow with the same rate
$N$.  This means that a family of systems with finite resources cannot be
modeled in this framework.

As shown in the proof of Proposition \ref{prop:ddSPN}, for an SPN in which
all places are covered by a P--semiflow, the number of tokens in a place is
bounded.  The minimal and maximal number of tokens in place $p_i$ will be
denoted by $ \mbox{MIN}(p_{i})$ and $ \mbox{MAX}(p_{i})$, respectively. If $S^{\BN}$ denotes the state space of the CTMC associated to the SPN $\mathcal{N}^{\BN}$, then
\[
 \mbox{MIN}(p_{i})=\min_{x \in S^{\BN}} \; x_{i},~
 \mbox{MAX}(p_{i})=\max_{x \in S^{\BN}} \; x_{i},
\]
where $x_{i}$ denotes the $i$--th component of the vector $x \in S^{\BN}$.
The following set is a
(possibly improper) superset of $S^{\BN}$
\begin{equation}
\label{eq:statespaceX}
\hat{S}^{\BN}=\left\{ x \in \mathbb{Z}^{n_{p}} : \forall i,\mbox{MIN}(p_{i})\leq x_{i} \leq
\mbox{MAX}(p_{i}) \mbox{ and } \forall \eta, x \nu^{(\eta)T} = \marc_{0}\nu^{(\eta)T} \right\}.
\end{equation}
which will help us to define the state space of the approximate models.

\subsection{From SPNs to ODEs}\label{ODE}

Having characterized, in Proposition~\ref{prop:ddSPN}, a family of SPNs
whose underlying CTMC family is nearly density dependent, we can apply the results
reported in Section~\ref{sec:ctmc2ode} to construct an approximation by
ODEs.  In order to provide a direct relation between the SPN family and the
ODE approximation, we construct the approximation not for the normalized
process $Z^{\BN}$, as it was done by Kurtz, but for the unnormalized
$X^{\BN}$.

The state of the system in the ODE approximation will be denoted by  $x^{\BN}(u) \in \mathbb{R}^{n_{p}}$.  A given infinite server transition $t_i$
moves ``fluid'' tokens in state $x^{\BN}(u)$ with speed
\begin{equation}
\label{eq:speed}
\sigma(t_i,x^{\BN}(u))=\lambda(t_i)\min_{j:I(p_j,t_i)\not=0}\frac{x^{\BN}_j(u)}{I(p_j,t_i)}
\end{equation}
which depends on the rate of the transition, $\lambda(t_i)$, and on its
{\em enabling degree}, calculated now at the ``fluid'' state.  The
number of tokens in the $i$-th place is then approximated by the following
ODE system:
\begin{equation}\label{eq:ode1}
dx^{\BN}_i(u)=\sum_{j=1}^{n_{t}} \sigma(t_j,x^{\BN}(u)) L(p_i,t_j) du
\end{equation}
so that if place $p_i$ is an input (output) place of transition $t_j$, then
transition $t_j$ is removing (adding) tokens from (to) place $p_i$
according to the current speed of the transition and the multiplicity given
by function $I$ ($O$).

It is easy to see that the ODE approximation maintains the same invariance
properties, expressed by the P--semiflows, that the original SPN enjoys.
Therefore the ODE system is redundant and each P--semiflow in the minimal
set can be used to derive one of its components from the others and hence
the system can be reduced by as many equations as the number of P--semiflows in the minimal set.

The fluid state space $\hat G^{\BN} \in \mathbb{R}^{n_{p}}$ of the deterministic
approximation $x^{\BN}$ is
\begin{equation}\label{fluidstatespace}
\hat G^{\BN}=\left\{ x \in \mathbb{R}^{n_{p}} : \forall i, \mbox{MIN}(p_{i})\leq x_{i} \leq
\mbox{MAX}(p_{i}) \mbox{ and }
\forall \eta, x {\nu^{(\eta)}}^T = \marc_{0}{\nu^{(\eta)}}^T \right\}
\end{equation}
that is the convex hull of the set $\hat{S}^{\BN}$ given in
eq. \eqref{eq:statespaceX}.  Theoretically, the fluid approximation may
visit all the points between the boundaries that are compatible with the
P--semiflows and the initial markings (gaps are filled).  In practice,
however, as the model follows a deterministic trajectory, only a $n_{p}$
dimensional curve is covered.  Moreover, if the initial marking belongs to
the interior of $\hat G^{\BN}$, denoted by $\mathring{G}^{\BN}$, where all the
components are strictly between the bounds, then the approximation
cannot visit the boundary defined as
\begin{equation}\label{eq:statespaceboundary}
\Gstar^{\BN} = \{ x \in \hat G^{\BN} : x_{i} =  \text{MAX}(p_{i}) \text{ or } x_{i}=\text{MIN}(p_{i}), \text{for at least one }i.  \}
\end{equation}


\subsection{From SPNs to SDEs}\label{SDE}

Under the hypothesis of density dependence, the procedure
illustrated in Section \ref{sec:sde} can be applied in this case too, and the number of
tokens in each place can be approximated by the diffusion process
$\Upsilon^{\BN}(u)=NY^{\BN}(u)$, with components $\Upsilon^{\BN}_{i}$ given by
the following system of SDEs:
\begin{align}\label{eq:sde}
d\Upsilon^{\BN}_{i}(u)  =
\sum_{j=1}^{n_{t}} \sigma(t_j,\Upsilon^{\BN}(u)) L(p_i,t_j) du
 + \sum_{j=1}^{n_{t}}\sqrt{ \sigma(t_j,\Upsilon^{\BN}(u)) }L(p_i,t_j) dW_{j}(u)
\end{align}
where each $W_{j}$ is an independent one-dimensional Brownian motion.

Just like the ODE approximation, also the diffusion approximation enjoys the
invariance properties present in the SPN.  Thus the size of the system can be reduced
 by removing one equation for each minimal  P--semiflow.

As already remarked, the diffusion approximation $Y^{\BN}$ introduced in
Section \ref{sec:sde} is defined only up to the first exit from an open set
and the same holds for $\Upsilon^{\BN}$.  The natural state space for the
process $\Upsilon^{\BN}$ would be the closed set $\hat G^{\BN}$ as defined in
eq. \eqref{fluidstatespace}.  However the diffusion approximation is only
valid up to the first exit from $\mathring{G}^{\BN}$. If the underlying CTMC is only nearly density dependent a suitable extension of approximation \eqref{eq:sde} can still be applied.

\subsection{From SPNs to SDEs in bounded domains}\label{sec:bd}

We have already stressed that for an SPN model satisfying the hypothesis of
Proposition \ref{prop:ddSPN} the number of tokens in each place is bounded
between a minimal and a maximal value. In many real examples the CTMC
visits some of the states corresponding to a minimal/maximal marking quite
often. From a modeling point of view this is not surprising: in Section
\ref{sec:exp_results} we revise the epidemiological
Susceptible-Infectious-Recovered (SIR) model where it is natural (at
least for some ranges of parameters) that the number of infected people
is zero for a non-negligible time. We claim that for this kind
of processes, both the fluid limit in eq. \eqref{eq:ode1} and the diffusion
process in eq. \eqref{eq:sde} fail to give satisfactory approximations of
the original process $X^{\BN}(u)$ for finite $N$. Numerical evidence of this
fact is given in Section \ref{sec:exp_results}, Table \ref{tab:exp1}.

An explanation for such a failure is that, when the (possibly improper)
state space $\hat{S}^{\BN}$ of the chain is embedded into the fluid state
space $\hat G^{\BN}$, each state where the marking of a place is minimal or
maximal is mapped to the boundary $\Gstar^{\BN}$.  While for finite $N$
the original CTMC may visit the ``boundary'' often and eventually stay
there for a long time, the deterministic fluid limit of Section \ref{ODE}
always remains in the interior.  On the other hand the diffusion
approximation of Section \ref{SDE} is valid only up to the first time it
leaves $\mathring{G}^{\BN}$ and hence cannot give a good approximation.

To overcome these difficulties, we abandon the ODE approach and propose
here a new approximation that improves the SDE method by carefully
introducing a suitable behavior at the boundaries with which we mimic the
original chain.  In the original CTMC, indeed, the system is described as a
pure jump process with discrete events which perfectly takes into account
the ``boundaries''.  The SDE approximation of Section~\ref{SDE} is a less
complex continuous model, but it does not deal with the behavior at the
boundaries. The new approximation we propose aims at taking the best from
both: time by time that part of the system which is not at the boundary is
still fluidified with the SDE approach and the rest, which involves the
boundaries, are kept discrete.  In particular the resulting process is a
jump--diffusion $\tilde\Upsilon^{\BN}(u)$ that we will describe  in
details.

If the process is initialized in $\mathring{G}^{\BN}$, the new approximating process $\tilde\Upsilon^{\BN}(u)$
evolves as the diffusion $\Upsilon^{\BN}(u)$ up to the first time it reaches $\Gstar^{\BN}$. From now on,  the set of places $P$  is (dynamically,
depending on the current value of $\tilde\Upsilon^{\BN}(u)$) split into those
that are at the boundary
\[
\Pstar = \{p_{i} \in P : \tilde\Upsilon^{\BN}_{i}(u)= \text{MAX}(p_{i}) \text{ or } \tilde\Upsilon^{\BN}_{i}(u)=\text{MIN}(p_{i})   \}
\]
and those that remain in the interior $\mathring{P}=P-\Pstar$.
At the same time we (dynamically) split the set $T$ of the transitions into
those that might move one of the components currently at the boundary
\begin{equation}\Tstar=\{t\in T: \exists p_{i} \in \Pstar\text{ such that }\ L(p_{i},t)\neq 0 \}\label{vecT}\end{equation}
and those in $\Tring=T-\Tstar$ that do not affect the places currently in
$\Pstar$.  As far as the transitions in $\Tstar$ do not fire, the
subsystem made of the places in $\mathring{P}$ and the transitions in $\Tring$
can still be approximated by a fluid SDE system (in the reduced state space
that does not include the components in $\Pstar$) whose equations are
analogous to eq. \eqref{eq:sde} except that the sums are
restricted to the transitions in $\Tring$. The transitions in $\Tstar$ cannot
be fluidified since they include the dynamics of the components at the
boundary. We keep them discrete and we encode them into a jump process
which is responsible for all the events of the type ``place $p_{i}$ leaves
the boundary''. The amplitudes and the intensities of the jumps are
formally taken from the original CTMC and depend on the complete state of
the process.  Finally, the approximating
jump-diffusion $\tilde\Upsilon^{\BN}(u)$ which embodies both the fluid
evolution and the discrete events solves the following system of SDEs
\begin{align}\label{eq:Jsde}
d \tilde\Upsilon^{\BN}(u) = \hspace{-2mm}
\sum_{j:\;t_{j} \in \Tring} \hspace{-1mm}L(t_{j})\left(\sigma(t_{j},\tilde\Upsilon^{\BN}(u)) du
 +\sqrt{ \sigma(t_j,\tilde\Upsilon^{\BN}(u)) } \; dW_{j}(u)\right)+ \hspace{-2mm}\sum_{i:\;t_{i} \in \Tstar} \hspace{-1mm}L(t_{i})dM^{\BN}_{i}(u)
\end{align}
where  $M^{\BN}_{i}(u)$ is the counting process that describes how many times transition $t_{i}$ has fired in the time interval $(0,u]$ and whose intensity is given by
\[
\mu_{i}(\tilde\Upsilon^{\BN}_{j}(u-))=\lambda(t_{i}) \min_{j:I(p_j,t_i)\not=0} \left\lfloor \frac
{\tilde\Upsilon^{\BN}_{j}(u-)}{ I(p_j,t_{i})} \right\rfloor
\]
which depends on the actual state of the process $\tilde\Upsilon^{\BN}$ right
before the jump.

Equation \eqref{eq:Jsde} has a component for each place in the net. This
might seem contradictory with the description we have given above according
to which only the places in $\mathring{P}$ are fluidified. Let us however
remark that the fluid increments in the first sum of equation
\eqref{eq:Jsde} do not affect the component at the boundary since if
$t_{j}\in \Tring$ then $L(t_{j},p_{i})=0$ for any $p_{i}$ in $\Pstar$:
their dynamics is solely affected by the jumps. On the other hand a
component that is not at the boundary at a given time can reach it due to the continuous compounding of the fluid increments that
sums up with the effect of the jumps.

Equation \eqref{eq:Jsde} is self-explanatory, but not rigorous since the
splitting of $T$ depends on the state and needs to
be updated continuously with the evolution of the system. According to the
definition \eqref{vecT}, at any time $u$ we can decide if a transition $t$
belongs to $\Tring$ by looking at the state $\tilde\Upsilon^{\BN}(u)$ of the
system and calculating the following quantity
\[
\theta\big(t,
\tilde\Upsilon^{\BN}(u)\big)=\sum_{i=1}^{n_p}\left|L(p_{i},t)\right| \cdot
\mathbb{1}\big(\tilde\Upsilon^{\BN}_{i}(u)= \text{MAX}(p_{i}) \text{ or }
\tilde\Upsilon^{\BN}_{i}(u)=\text{MIN}(p_{i})\big)
\]
where $\mathbb{1}(\cdot)$ is the indicator function of the event in
parentheses.

If $\theta\big(u, \tilde\Upsilon^{\BN}(u)\big)=0$ then $t\in
\Tring$. Accordingly, equation \eqref{eq:Jsde} can be recast into the following
form

\begin{align}\label{eq:Jsde}
d \tilde\Upsilon^{\BN}(u) =& \hspace{-1mm}
\sum_{j=1}^{n_{t}} \mathbb{1}\big\{\theta\big(t_{j}, \tilde\Upsilon^{\BN}(u)\big)=0\big\}\: L(t_{j})\left(\sigma(t_{j},\tilde\Upsilon^{\BN}(u)) du +\sqrt{ \sigma(t_j,\tilde\Upsilon^{\BN}(u)) } \; dW_{j}(u)\right)\notag\\
+& \hspace{-1mm}\sum_{j=1}^{n_{t}} \mathbb{1}\big\{\theta\big(t_{j}, \tilde\Upsilon^{\BN}(u)\big)\neq0\big\}\: L(t_{j})\:dM^{\BN}_{j}(u).\notag
\end{align}

\begin{algorithm}[tbp]
\caption{Algorithm for solving SDE systems}
\label{algo1}
\begin{algorithmic}[1]
\Function{SolveSSDE}{$\emph{SSDE},\emph{step},\emph{MaxRuns},\emph{FinalTime}$}
\footnotesize
\Statex  $\emph{SSDE}$ = SDE system.
\Statex  $\emph{step}$ = step used in the Euler-Maruyama solution.
\Statex  $\emph{MaxRuns}$ = maximum number of runs.
\Statex $\emph{FinalTime}$ = maximum time  for each run.
\Statex $\emph{Value}$= a matrix encoding for each run the SDE value at the current step.
\footnotesize
\State \emph{run}=1;
\State \emph{SSDE.Init(Value)};\label{line1}
\While {($\emph{run} \leq \emph{MaxRuns})$}
	\State \emph{u} = 0.0;
	\While {($\emph{u} \leq \emph{FinalTime})$}\label{line2}
		\State \emph{Value[run].Copy(PrValue)};
		\State \emph{$h$} = step;
	    \State \emph{$\mathring{P}$ = SSDE.splitP(Value)};\label{line3}
	    \State \emph{$\langle \Tstar,\Tring \rangle$=$\mathring{P}$.SplitT()};\label{line4}
		\State \emph{$\langle t,h \rangle$ = $\Tstar$.Jump(step)};\label{line5}
		\For{($\emph{SDE} \in \emph{SSDE}$)}
			\State $\emph{SDE.computeJump(t,$h$,Value[run],PrValue)}$;
			\State $\emph{SDE.computeEuler-Maruyama(Value[run],PrValue,$\Tring$,$h$)}$;
		\EndFor
		\State \emph{SSDE.Norm(Value)};
		\State \emph{u} += $h$;
	\EndWhile{}\label{line6}
	\State \emph{run}++;
\EndWhile
\State \textbf{return} \emph{Value}.distribution();
\EndFunction
\end{algorithmic}
\end{algorithm}

Now we provide in Algorithm~\ref{algo1}  a pseudo-code which describes the implementation of the solution of this  approximating  jump-diffusion,  which extends the standard Euler-Maruyama method~\cite{kloeden1992} in order to solve SDEs  in which  fluid evolution and  discrete events coexist.
We assume that the reader is familiar with  the standard Euler-Maruyama method; for an  its complete description the reader can refer to~\cite{kloeden1992}.

The algorithm takes in  input the model definition represented by the system of SDEs (i.e., \emph{SSDE}) corresponding to  Eq.~\eqref{eq:Jsde},
the maximum
Euler-Maruyama step (i.e., \emph{step}), the maximum number of runs (i.e.,
\emph{MaxRuns}), and a final time (i.e., \emph{FinalTime}) for which the
solution is computed; and it returns the distribution of the variables
at~\emph{FinalTime} time.  A floating point matrix (i.e. \emph{Value}) of
dimensions $\emph{MaxRuns} \times |\emph{SSDE}|$ is used to store for each
run the current value of the variables.  In line~\ref{line1} the method
\emph{init()} initializes the matrix \emph{Value} according to the initial
state, so that each run will start from the same initial values.  Then for
each run, the SDEs are recursively solved until the final solution time
(i.e., \emph{FinalTime}) is reached (from line ~\ref{line2}
to~\ref{line6}).  As previously described the SDEs solution at time $u$
requires some steps: first the set of places is split in the two sub-sets
$\mathring{P}$ and $\Pstar$ (i.e. line~\ref{line3}, method
\emph{SplitP()}), secondly the sets $\Tstar$ and $\Tring$ are computed
(i.e. line~\ref{line4}, method \emph{SplitT()}).  After that the method
\emph{Jump()} in line~\ref{line5} is called to model the jump process which
is responsible for all the events moving a place outside its
boundaries. Its output is a tuple $ \langle t,h\rangle $ which specifies
which transition $t\in\Tstar$ will fire at time $u+h$ with $h\leq
step$. Obviously in a time interval $(u,u+step)$ no discrete events could
happen (even if $\Tstar\neq \emptyset$), in this case the method
\emph{Jump()} returns the tuple $ \langle -1,step\rangle $.  Then, for each
equation in \emph{SSDE} the method \emph{computeJump()} updates the
variables considering the discrete event modeled by $t$; while the method
\emph{computeEuler-Maruyama()} updates the variables considering the fluid
evolution accounting  only for the events in $\Tring$.  In the end of each
step the method \emph{Norm()} is called to normalize all the new computed
variables taking into account the P-invariants.  Finally the method
\emph{distribution()} generates from the final value of the variables the
distribution of all the involved quantities.

\section{Experimental results}\label{sec:exp_results}

\begin{figure}[tbp]
   \centering
   \includegraphics[width=1.0\textwidth]{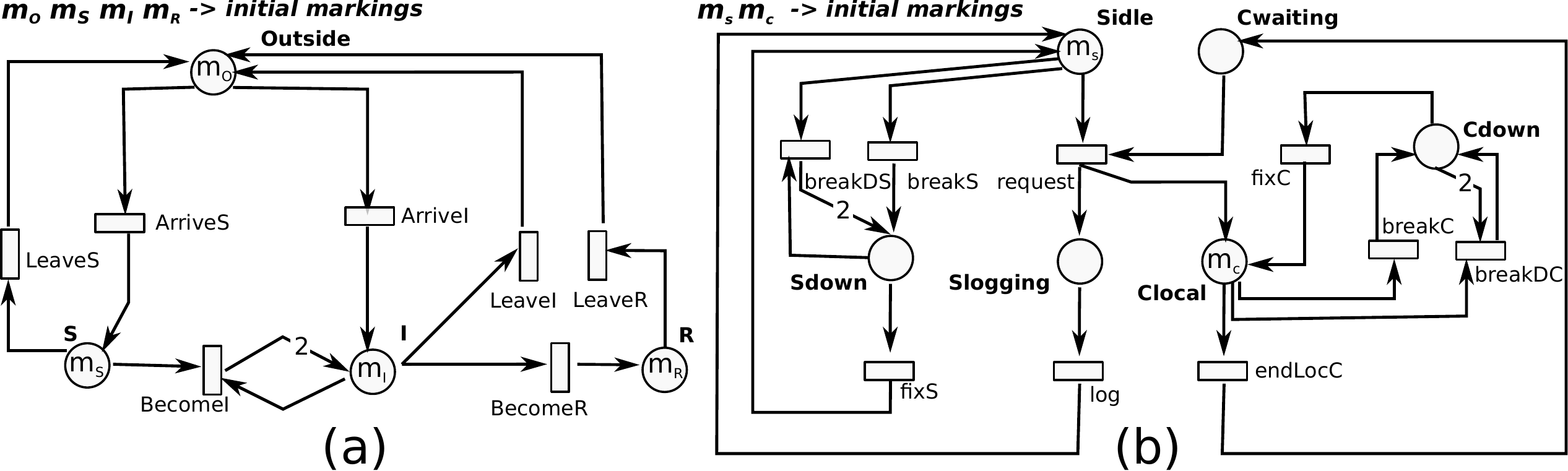}
   \caption{(a) SPN net inspired by SIR model;\ \ \ \  (b) SPN model inspired by PEPA client-server model}\label{Fig:Model}
\end{figure}

In this section we report the results obtained from the analysis of  two SPN models, to show the quality and the robustness of the approximations obtained with our new approach.
The first model depicted in Fig.~\ref{Fig:Model}(a)  is inspired by the  epidemiological Susceptible-Infectious-Recovered (SIR) mathematical representation of this problem originally introduced in~\cite{SIR27} and discussed with respect to some of its variants in~\cite{Beccuti12,Beccuti13}.
It describes the diffusion of an epidemic on a large population, and assumes  that the population members  are part of three sub-populations according to their health status: (a) \emph{susceptible members} (represented in the model by tokens in place \emph{S}) that  are not ill, but that are susceptible to the disease; (b) \emph{infected members} (i.e., tokens in place \emph{I}) that are subject to the disease and can spread it among susceptible members; (c) \emph{recovered members} (i.e., tokens in place \emph{R}) that were previously ill and   are now immune.
Each member of the population typically progresses from susceptible to infectious (i.e., firing of transition \emph{BecomeI}) depending on   the number of  infected members, and  from infectious to recovered (i.e., firing of transition \emph{BecomeR}).
Moreover,  members of the population can leave the infected area (i.e., firing of transitions  \emph{LeaveS}, \emph{LeaveI} and  \emph{LeaveR}) to reach the ``outside world'' (i.e., place \emph{Outside}) and  members of the outside world can enter into the infected area joining the sub-population of the susceptible (i.e., transition \emph{ArriveS}) or infected  members (i.e., transition \emph{ArriveI}).

The second model represented in Fig.~\ref{Fig:Model}(b) is derived from the the PEPA representation of a Client-Server  system studied in~\cite{JANE}.
It describes an environment where servers, initially idle (tokens in  place \emph{Sidle}), are waiting for a client synchronization (represented by transition \emph{request}).
When the synchronization is completed, the client returns to its local computation
(i.e., place \emph{Clocal}) until the occurrence of the next synchronization  (i.e., transition \emph{endLocCl}), while the server executes an action log before becoming idle again (i.e., transition \emph{log}).
An idle server or a client in local computation may fail due to a virus infection (i.e., transitions \emph{breakS} and \emph{breakC}). Moreover an infected server or client can also infect other machines (i.e., transition \emph{breakDS} and \emph{breakDC}).
Finally a server or a client  recovers only when  an anti-malware software discovers the virus (transitions \emph{fixS} and \emph{fixC}).
\begin{figure}[tbp]
   \centering
   \includegraphics[width=0.45\textwidth]{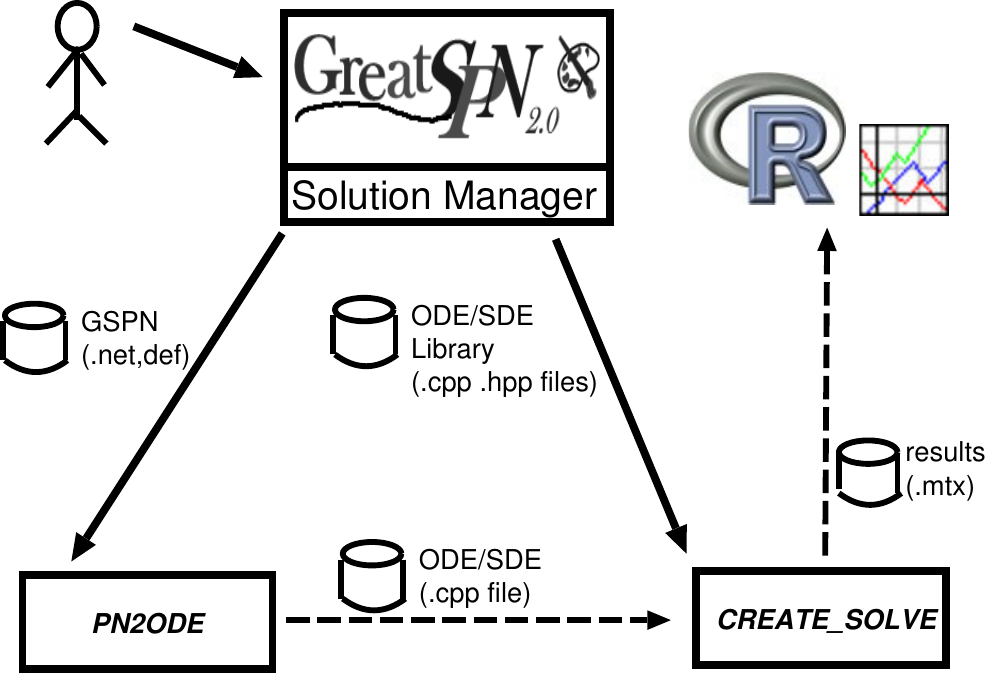}
   \caption{Framework architecture.}\label{Fig:Arc}
\end{figure}
All the experiments  performed on these two models have been carried on with a prototype implementation integrated  in GreatSPN framework~\cite{BabarBDM10}, which allows the generation of the ODE/SDE system from an SPN model and then the computation of its solution.
The architecture of this prototype is depicted in Fig.~\ref{Fig:Arc} where the framework components are presented by
rectangles, the component invocations by solid arrows, and the models/data exchanges by dotted arrows.
More specifically, GreatSPN is used to design the SPN model and  to activate the solution process, which comprises the following two steps:\\
1. \emph{PN2ODE} generates from an SPN model a  C++ file implementing the corresponding ODE/SDE system;\\
2. \emph{CREATE\_SOLVE}  compiles the previously generated C++ code  with the library implementing the SDE/ODE solvers, and executes it.\\
Finally the computed results are processed through the R framework to derive statistical information and  graphics. All the results have been obtained running these programs
on a 2.13 GHz Intel I7 processor with 8GB of RAM.

\begin{table}
\centering \scriptsize
\begin{tabular}{|c|r|r||c|r|}
\hline
 \multicolumn{3}{|c||}{\textbf{SIR model}} & \multicolumn{2}{|c|}{\textbf{Client Server model}}\\
  \hline
	 \multicolumn{1}{|c|}{\textbf{Transition}} &
         \multicolumn{1}{c|}{\textbf{rate (1$^\circ$ exp.)}}  &\multicolumn{1}{c||}{\textbf{rate (2$^\circ$ 3$^\circ$ exp.)}}&
          \multicolumn{1}{|c|}{\textbf{Transition}} &
         \multicolumn{1}{c|}{\textbf{rate}}\\
\hline\hline
\emph{ArriveS} & 0.5& 1.0& \emph{log} & 12\\
 \hline
\emph{ArriveI} & 0.5& 0.01& \emph{request} & 1\\
 \hline
\emph{BecomeI} & 1.0& 1.0& \emph{endLocC} & 0.2\\
\hline
\emph{BecomeR} & 0.5& 0.5& \emph{breakS}, \emph{breakC} & 0.0007 0.00002\\
\hline
\emph{LeaveS}, \emph{LeaveR} & 0.02& 0.02& \emph{breakDS}, \emph{breakDC} & 0.8 1.4\\
\hline
\emph{LeaveI} & 0.1& 0.1&\emph{fixS}, \emph{fixC} & 0.001 0.001\\
\hline
\end{tabular}
\vspace{0.2cm}
 \caption{Transition rates.}\vspace{-0.8cm}
\label{tab:rate}
\end{table}

In the first set of experiments performed on the model of Fig.~\ref{Fig:Model}(a), we consider a  situation in which the effect of barriers does not influence the correctness/quality  of  the solution computed by ODEs, so that we are able to properly compare the ODE approach with our new one. For this purpose, we assume that the initial marking of the model corresponds to a total of
200 people equally distributed over all the places of the model ($m_{O} = m_{S} = m_{I} = m_{R} = 50$)\footnote{This choice is done to avoid the initial barrier effect due to empty places.} and that the transition rates are chosen as reported in the second column of Table~\ref{tab:rate}.
Observe that we consider only 200 people since we want to compare the results obtained by the ODE and SDE approaches with those derived by solving the CTMC underlying the same SPN model\footnote{The generation of a CTMC from a SPN model and its solution are obtained using the GSPN solvers  available in  the GreatSPN suite.}, and because we want to stress the reliability of these methods even for cases of large, but not infinite, populations as it would instead be requested by Kurtz's theorem for ensuring the accuracy of the approximations.
\begin{figure}[h]
\subfigure[Temporal behavior of the  SIR members computed solving ODEs.]{
 \includegraphics[width=0.45\textwidth]{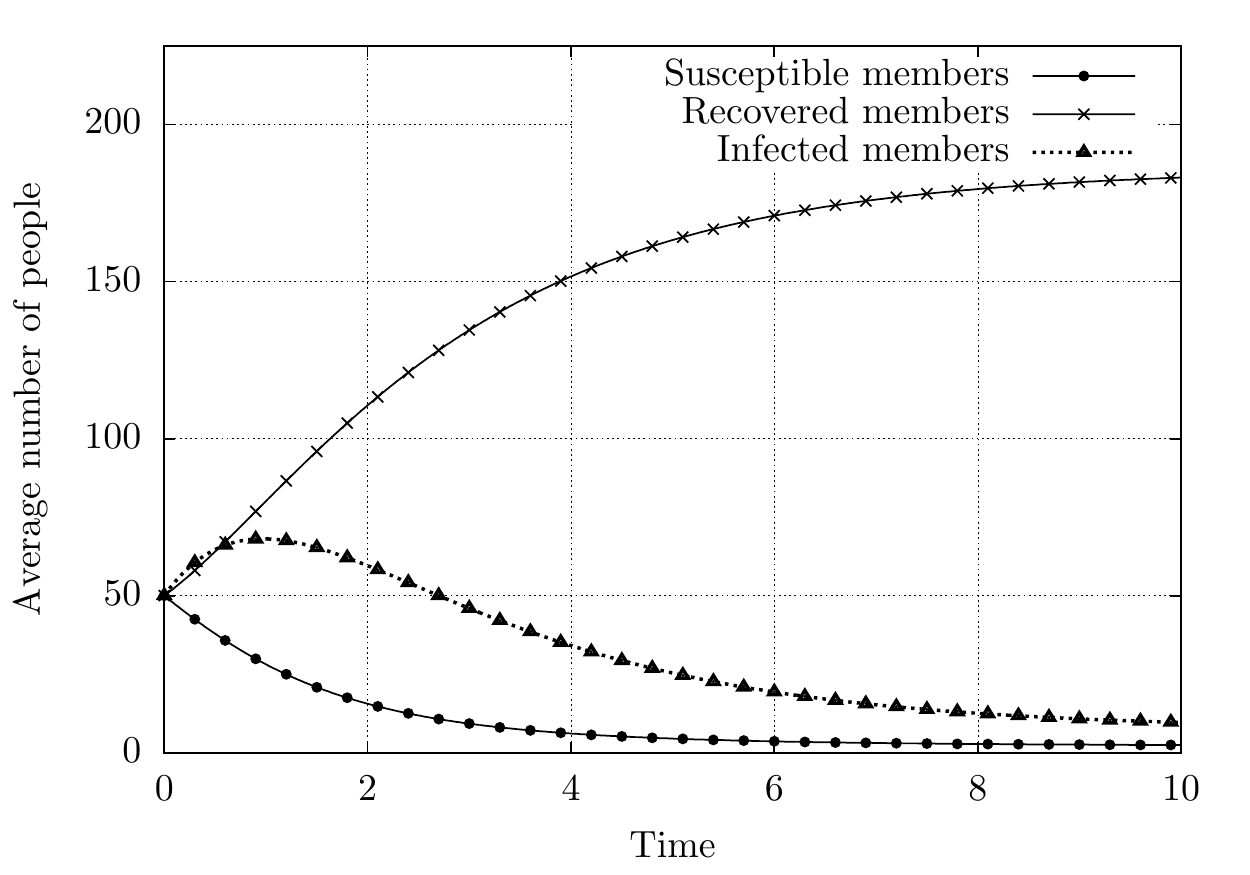}}
\subfigure[Comparison between SDE traces (solid lines) and ODE one (dashed line)]{
 \includegraphics[width=0.45\textwidth]{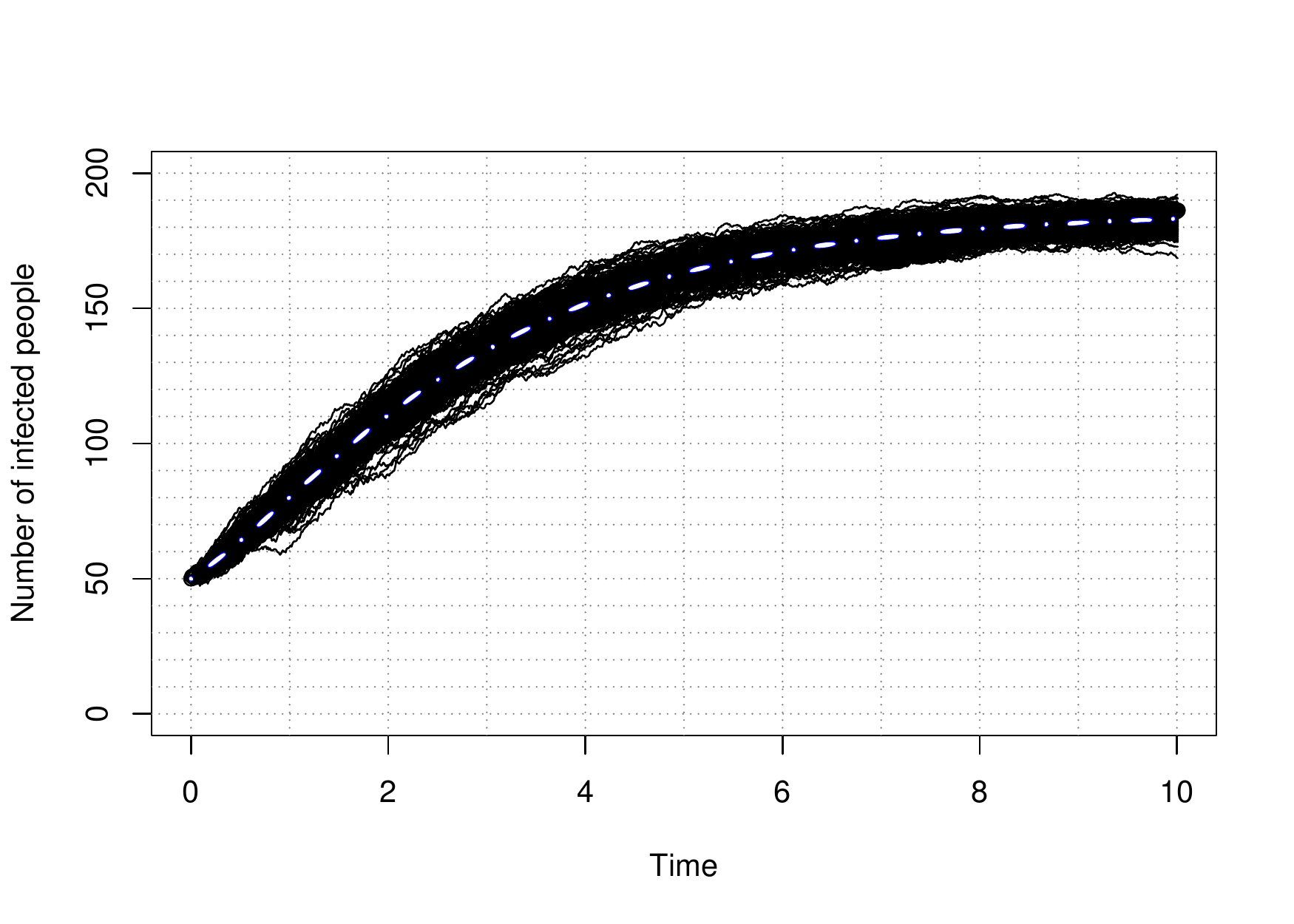}}
\\
\subfigure[Comparison between average SDE traces (black line) and   ODE one  (red line)]{
\includegraphics[width=0.45\textwidth]{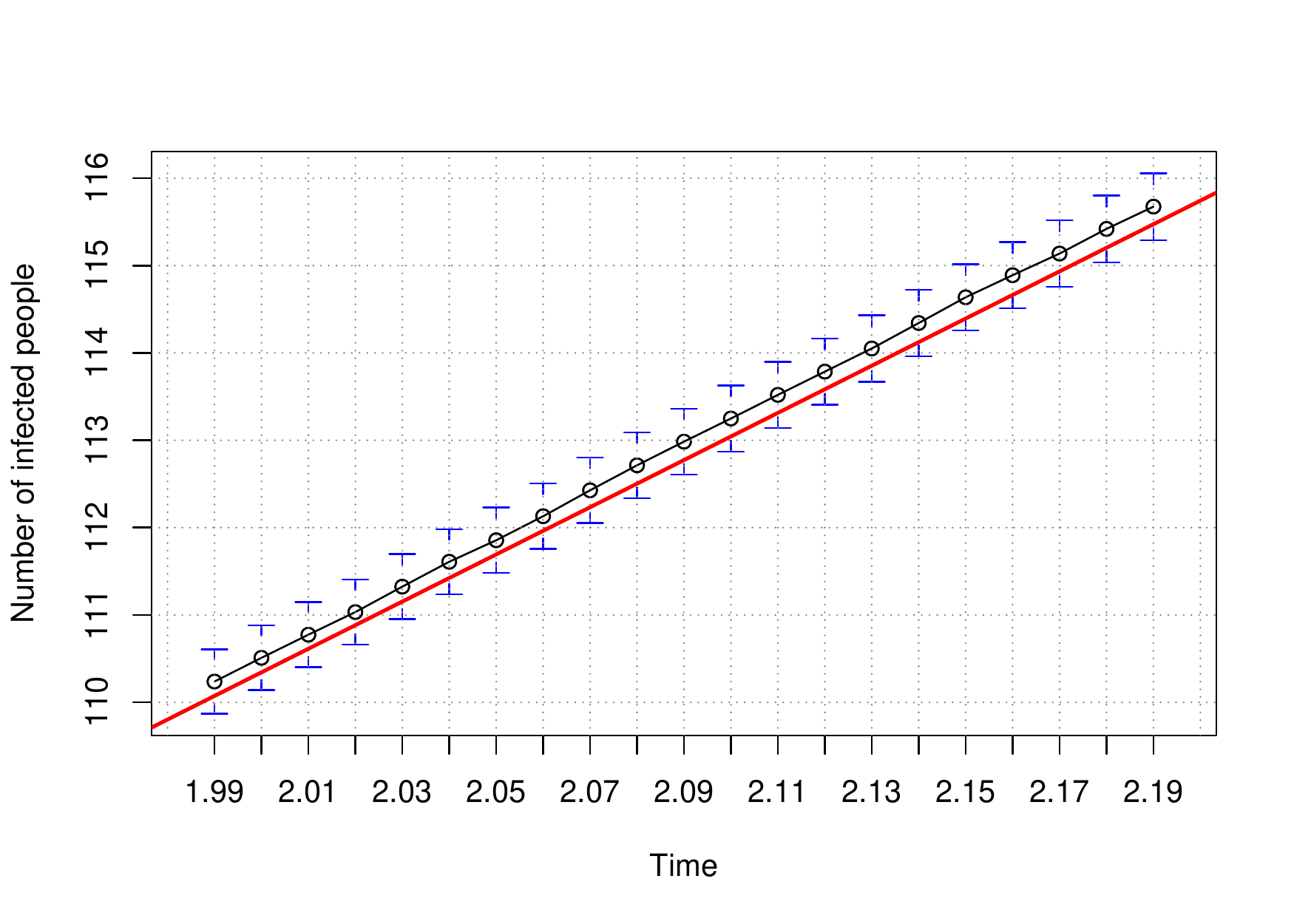}}
\subfigure[Comparison between SDEs and CTMC for infected members.]{
 \centering
   \includegraphics[width=0.45\textwidth]{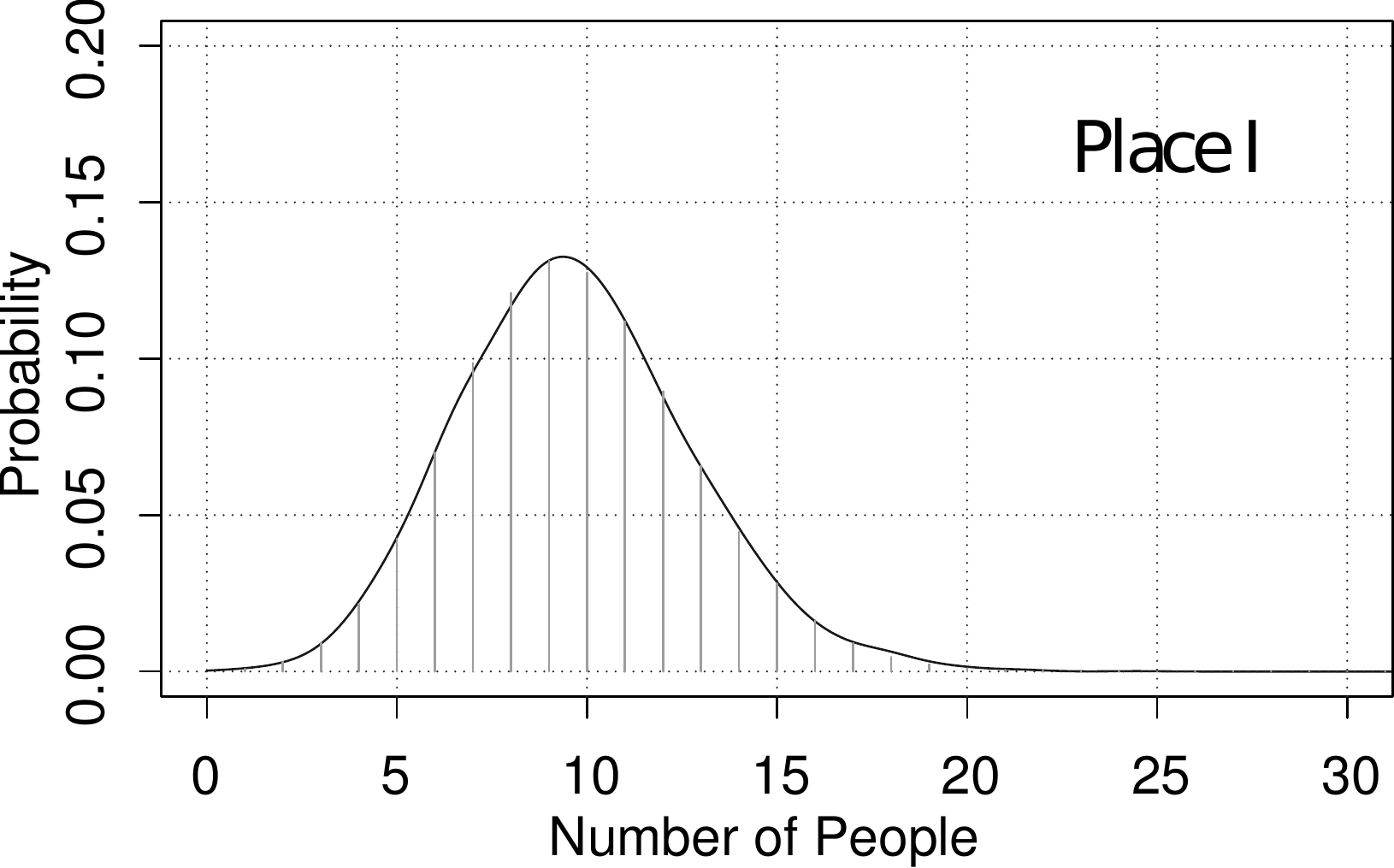}}
\caption{First SIR experiment}\vspace{-0.5cm}\label{Fig:2}
\end{figure}
Fig.~\ref{Fig:2}(a)  shows the temporal behavior (i.e., between 0 to 10 time units) of the  SIR members computed solving  the ODE system. Similar figures obtained from the transient solution of the CTMC at different time points are  not explicitly reported on these diagrams, but ensure the reliability of the method in this case. Figs.~\ref{Fig:2}(b) and~(c) show instead the temporal behavior of the  Infected members as it is obtained by solving the SDE system with 1000 runs and Euler step equal to 0.001. In details, Fig.~\ref{Fig:2}(b)   plots  the SDE traces (black lines) with respect to  the ODE trace that is represented as the white line in the middle of the tick cloud of black lines; Fig.~\ref{Fig:2}(c) reports the average SDE trace computed at specific points in time by considering the values obtained from the Euler simulations and augmented by confidence intervals that support the reliability of the method.  The averages of the SDE results are represented in the diagram with the dashed black line. The results obtained with the ODEs are represented instead with a solid red line that lies next to the previous one and that is covered (in any case) by the confidence intervals obtained from the Euler Simulation.
Figs.~\ref{Fig:2}(b) and (c) show that the solution quality of both approaches are comparable, but  it is important to highlight that the execution time for the two approaches are quite different.  Indeed the solution for the ODE system requires $\sim 1$sec., while that for the SDE model is six times slower ($\sim 6$sec.). However, we can observe that this overhead in the SDE solution can be  justified by the fact that the  SDEs provide also the probability distribution of each sub-population at  desired specific times.
For instance, Fig.\ref{Fig:2}(d) compares the probability distribution of Infected members at time 10 computed with GreatSPN solving the  CTMC underlying the same SPN (grey dotted bars) with that computed by the SDEs (black lines) with 1000 runs.
A good approximation is obtained with the SDE  approach reducing the execution time by a factor of $\sim 30$  with respect to that of the CTMC solution.

The second set of experiments  addresses instead a case where the presence  of barriers has a negative effect on the  correctness/quality  of  the solution computed by ODEs, while it is handled correctly by our SDE solution algorithm  which is still able to reproduce the expected  of the model. To stress this result, we performed this new set of experiments using the same basic model discussed before, where we changed the transition rates as reported  in the third column of Table~\ref{tab:rate} and we assumed that all the population members were originally concentrated in place \emph{Outside} ($m_{O}=200$ and $m_{S} = m_{I} = m_{R} = 0$). With this configuration of the model, we can observe that  the temporal behavior of the  SIR members derived by solving the ODE system disagree with those computed solving the  CTMC (see Table~\ref{tab:exp1} second and third columns).
This does not contradict Kurzt's theorem since it cannot be directly applied due to the choice of these transition rate values which leads the number of Infected members to be equal to $0$ (i.e., corresponding to  the lower bound of this quantity)  most of the time. Instead our approach based on SDE is still able to cope with this case providing a good approximation for the CTMC solution (i.e., Table~\ref{tab:exp1} second and fourth columns).
\begin{table}
\centering \scriptsize
\begin{tabular}{|c|r|r|r|}
  \hline
	 \multicolumn{1}{|c|}{\textbf{Sub-population}} &
         \multicolumn{1}{c|}{\textbf{CTMC}}  &\multicolumn{1}{c|}{\textbf{ODE}}&\multicolumn{1}{c|}{\textbf{SDE}}\\
\hline\hline
\emph{S} &62.636& 4.367& 61.906 +/-1.55\\
 \hline
\emph{I} &127.965 & 183.825& 128.780 +/-1.48\\
 \hline
\emph{R} &4.280 & 4.454 &4.268 +/-0.04\\
 \hline
\end{tabular}
\vspace{0.2cm}
 \caption{Average number of members in each sub-population at time 100}\vspace{-0.8cm}
\label{tab:exp1}
\end{table}

In Fig.\ref{Fig:comp1} the probability distributions  of SIR members at time 100 derived by the CTMC (grey dotted bars) are compared  with those  computed on the same model by SDE (black lines) with 5000 runs and step 0.01.
From these graphs it is clear how our approach is still able to reproduce with a high precision  these complex probability distributions reducing the  memory demand (from $\sim$162MB to $\sim$1MB) and the execution time (from $\sim$600s.  to $\sim$28s.).

\begin{figure}[tbp]
   \centering
   \includegraphics[width=0.90\textwidth]{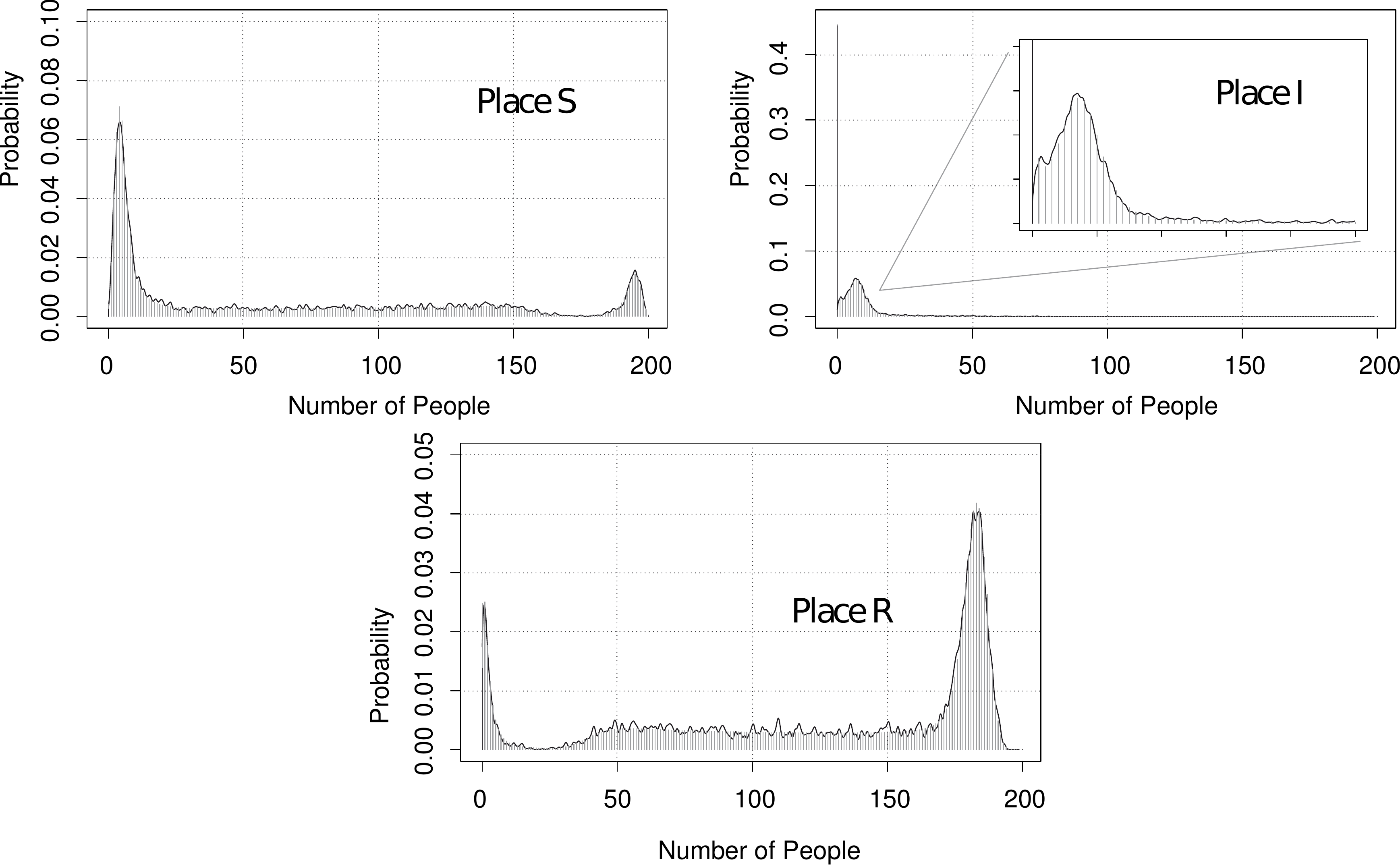}
   \caption{Second SIR experiment: Comparison between SDEs and CTMC for  SIR members}\vspace{-0.5cm}\label{Fig:comp1}
\end{figure}

 \begin{table}
\centering \scriptsize
\begin{tabular}{|r|r|r|r|r|}
  \hline
	 \multicolumn{1}{|c|}{} & \multicolumn{2}{c|}{\textbf{SDE}}& \multicolumn{2}{c|}{\textbf{Simulation}}\\
	  \multicolumn{1}{|c|}{\textbf{$|$Population$|$}}&\multicolumn{1}{c|}{\textbf{Exec. Time}}&\multicolumn{1}{c|}{\textbf{$E[I]_{T=100}$}}&\multicolumn{1}{c|}{\textbf{Exec. Time}}&\multicolumn{1}{c|}{\textbf{$E[I]_{T=100}$}}\\
 \hline
 \hline
 200 & 128.78 +/-1.48& 26s. &128.02+/-1.48 &12s.\\
  \hline
 2,000 & 73.78 +/-0.19 &26s. &73.51+/-0.19 & 155s. \\
  \hline
 20,000 &735.31+/-0.62 &27s. &735.03+/-0.61& 20m.\\
  \hline
 200,000 & 7352.62+/-0.62& 27s. &7352.32+/-0.62& 5h. \\
 \hline
\end{tabular}
\vspace{0.2cm}
 \caption{Comparing SDE approach with simulation varying the population size}\vspace{-0.8cm}
\label{tab:exp2}
\end{table}

Finally, to show how our approach scales when increasing the population, we report in Table~\ref{tab:exp2}. the results obtained with our method and those computed with (standard) Discrete Event Simulation.
The first column of  Table~\ref{tab:exp2} reports the
population size, the second and third show the average number of infected members at time 100 computed with the SDE approach and the execution time needed for its computation; similarly,  the last two columns contain the same information referred to the Discrete Event Simulation.
From these results it is clear how our approach is able to obtain a speed-up with respect to the Discrete Event Simulation requiring the same used memory.
Obviously this speed-up depends on the characteristic of the model and increases/decreases proportionally to   the time spent  by the  quantities  of interest  on their boundaries. Indeed, as explained in Sec.\ref{sec:SPNfluid}, every time that a quantity of interest reaches one of its  bounds, our approach uses a classical Discrete Event Simulation to derive its next evolution step.

 The last set of experiments is related to the Client Server system and shows how the SDE approach deals with models exhibiting
\emph{multi-modal} behavior.
Indeed, choosing  the transition rates as reported in the last two columns of Table~\ref{tab:rate}  and an initial marking with 120 idle servers and 10,000 clients in local computation, we can observe that the  probability distribution of tokens in place \emph{Cwaiting} at time 18, computed thought Discrete Event Simulation,  has a multi-modal shape (see Fig.\ref{Fig:ClSe4} dashed light line).
In particular, the first mode corresponds  to the situation where most of the clients failed,   the second one  to the situation where only few servers are down and few clients are failed, and the last mode to the situation where most of the servers are down and only few clients are failed.
Comparing  this dashed line with the black solid line plotting the same measure computed by the SDE approach, we can observe that a good level of approximation is attained by such an approach reducing the computation time by a factor of $\sim 13$.
Indeed, the two lines are often difficult to distinguish (thus showing the good level of agreement between the two methods) thus identifying the modes (picks) of the distributions in a satisfactory manner, even if the absolute values of the two lines are not perfectly matched in these cases.

\begin{figure}[tbp]
   \centering
   \includegraphics[width=0.50\textwidth]{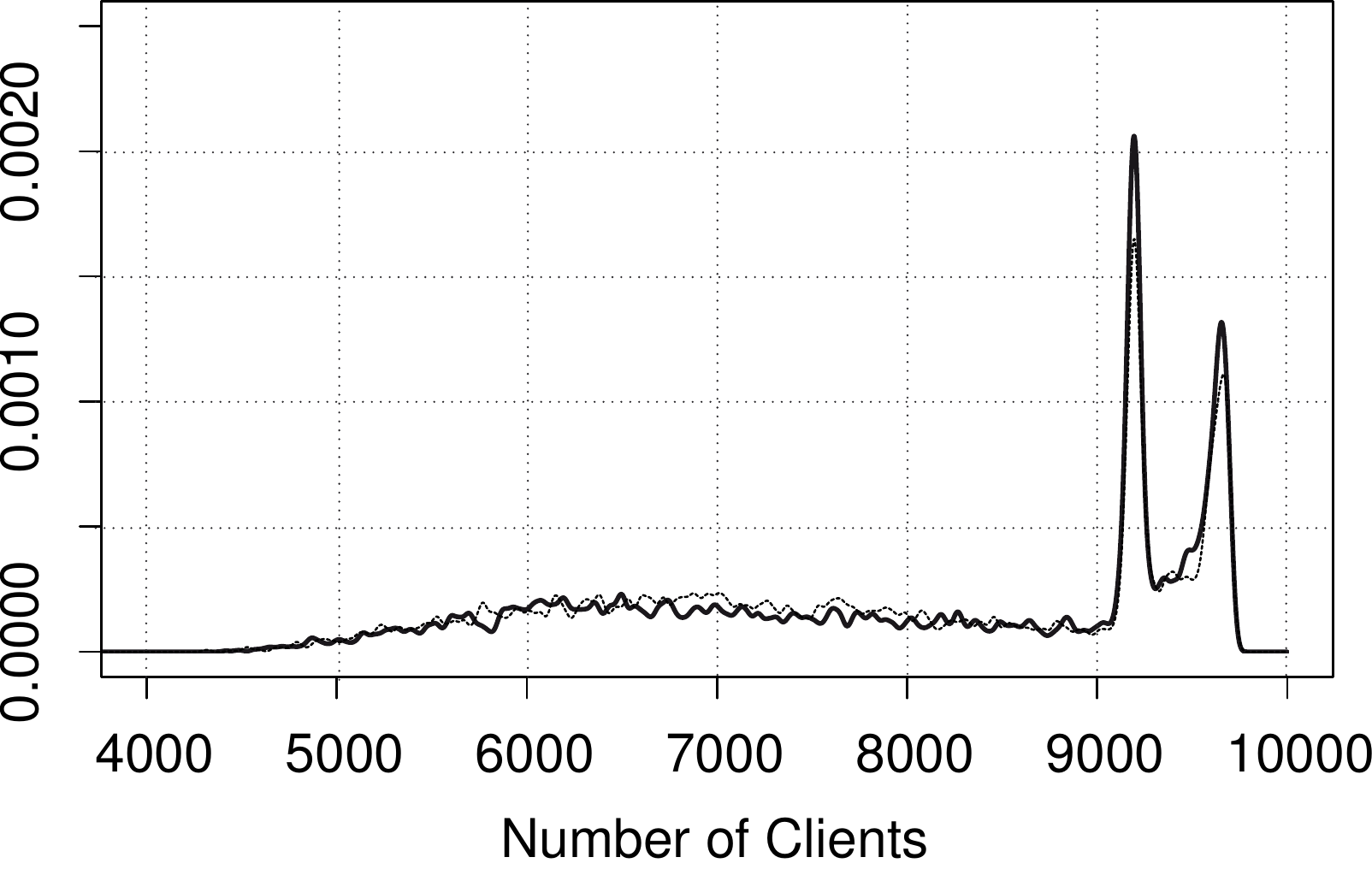}
   \caption{First client server experiment: comparison between SDEs and Discrete Event Simulation for the distribution of tokens in place \emph{Cwaiting} at time 18}\vspace{-0.5cm}\label{Fig:ClSe4}
\end{figure}


\section{Conclusion}

In this paper we considered approximations of SPNs.  We identified a class
of SPNs for which the underlying stochastic process is a density dependent
CTMC.  Consequently, it is possible to apply to these SPNs both a
deterministic approximation based on ODEs and a diffusion approximation
based on SDEs, both being introduced by Kurtz.  The diffusion
approximation, as presented in its original form, is defined only up to the
first exit from an open set.  Since in many applications barriers are
important, we extended the diffusion approximation by jumps that take into
account the behavior of the process on the barriers.  We showed by
numerical examples that the resulting jump diffusion approximation provides
precise information about the distribution of the involved quantities and
outperforms the deterministic approach when the original process is with
significant variability or it involves distributions whose mean value does
not carry much information.
Future work will investigate the numerical aspects of the prototype implementation used for the experiments and the possibility of applying these idea to the analysis of hybrid models.

{\footnotesize
\noindent
{\bf Acknowledgments.}
We gratefully acknowledge  the anonymous referees for their useful suggestions that helped improving the final version of the paper. This work has been supported in part by project ``AMALFI 
'' sponsored by Universit\aa\ di Torino and Compagnia di San Paolo. 
}

\bibliographystyle{splncs03}
\bibliography{bibl}

\begin{thebibliography}{10}
\providecommand{\url}[1]{\texttt{#1}}
\providecommand{\urlprefix}{URL }

\bibitem{BOABCDF95}
Ajmone~Marsan, M., Balbo, G., Conte, G., Donatelli, S., Franceschinis, G.:
  {Modelling with Generalized Stochastic Petri Nets}. J. Wiley, New York, NY,
  USA (1995)

\bibitem{BabarBDM10}
Babar, J., Beccuti, M., Donatelli, S., Miner, A.S.: Greatspn enhanced with
  decision diagram data structures. In: Proceedings of Applications and Theory
  of Petri Nets, 31st Int. Conference, Braga, Portugal, June 21-25,. pp.
  308--317. IEEE Computer Society (June 2010)

\bibitem{Bal-01}
Balbo, G.: {Introduction to Stochastic Petri Nets}. In: Brinksma, E., Hermanns,
  H., Katoen, J.P. (eds.) {Formal Methods and Performance Analysis, LNCS Vol.
  2090}, pp. 84--155. Springer-Verlag, Berlin, Germany (May 2001)

\bibitem{Beccuti12}
Beccuti, M., Franceschinis, G.: Efficient simulation of stochastic well-formed
  nets through symmetry exploitation. In: Proceedings of the Winter Simulation
  Conference. pp. 296:1--296:13. WSC '12, IEEE Computer Society (Dec 2012)

\bibitem{Beccuti13}
Beccuti, M., Fornari, C., Franceschinis, G., Halawani, S.M., Ba-Rukab, O.,
  Ahmad, A.R., Balbo, G.: From symmetric nets to differential equations
  exploiting model symmetries. The Computer Journal  (2013)

\bibitem{bremaud1999}
Br{\'e}maud, P.: Markov chains, Texts in Applied Mathematics, vol.~31.
  Springer-Verlag, New York (1999), gibbs fields, Monte Carlo simulation, and
  queues

\bibitem{Fish1978}
Fishman, G.S.: Principles of Discrete Event Simulation. John Wiley \& Sons,
  Inc., New York, NY, USA (1978)

\bibitem{GaetaTSE96}
Gaeta, R.: {Efficient Discrete-Event Simulation of Colored Petri Nets}. IEEE
  Transactions on Software Engineering  22(9),  629--639 (1996)

\bibitem{gillespie}
Gillespie, D.T.: The chemical langevin equation. J. Chem. Phys.  113,  297
  (2000)

\bibitem{SIR27}
Kermack, W., McKendrick, A.: A contribution to the mathematical theory of
  epidemics. Proceedings of the Royal Society of London. Series A  115(772),
  700--721 (Aug 1927)

\bibitem{klebaner}
Klebaner, F.C.: Introduction to stochastic calculus with applications. Imperial
  College Press, London, third edn. (2012)

\bibitem{kloeden1992}
Kloeden, P.E., Platen, E.: Numerical solution of stochastic differential
  equations, vol.~23. Springer (1992)

\bibitem{Ku70}
Kurtz, T.G.: Solutions of ordinary differential equations as limits of pure
  jump {M}arkov processes. Journal of Applied Probability  1(7),  49--58 (1970)

\bibitem{kurtz1976limit}
Kurtz, T.G.: Limit theorems and diffusion approximations for density dependent
  markov chains. In: Stochastic Systems: Modeling, Identification and
  Optimization, I, pp. 67--78. Springer (1976)

\bibitem{kurtz1978strong}
Kurtz, T.G.: Strong approximation theorems for density dependent markov chains.
  Stochastic Processes and Their Applications  6(3),  223--240 (1978)

\bibitem{molloy:spn}
Molloy, M.K.: Performance analysis using stochastic {P}etri {N}ets. IEEE
  Transactions on Computers  31(9),  913--917 (1982)

\bibitem{JANE}
Pourranjbar, A., Hillston, J., Bortolussi, L.: Don{\textquoteright}t just go
  with the flow: Cautionary tales of fluid flow approximation. In: Tribastone,
  M., Gilmore, S. (eds.) EPEW 2012, and UKPEW 2012. Lecture Notes in Computer
  Science, vol. 7587, p. 156{\textendash}171. Springer (2012)

\bibitem{Tr10}
Tribastone, M.: Scalable differential analysis of large process algebra models.
  In: 7th Int. Conference on the Quantitative Evaluation of Systems. p. 307.
  IEEE Computer Society, Williamsburg, Virginia, USA (sept 2010)

\bibitem{TrGiHi12}
Tribastone, M., Gilmore, S., Hillston, J.: Scalable differential analysis of
  process algebra models. IEEE Trans. Software Eng.  38(1),  205--219 (2012)

\end{thebibliography}

\end{document}